\documentclass{amsart} 

\usepackage{amsmath}
\usepackage{amssymb}
\usepackage{amsthm}
\usepackage{color}
\usepackage{epsfig}
\usepackage{mathtools} 
\usepackage[nice]{units}

\newcommand{\be}[1]{\begin{equation}\label{#1}}
\newcommand{\ee}{\end{equation}}

\theoremstyle{plain}
\newtheorem{Proposition}{Proposition}

\newtheorem{Lemma}[Proposition]{Lemma}
\newtheorem{Theorem}[Proposition]{Theorem}

\theoremstyle{definition}
\newtheorem{Definition}{Definition}
\newtheorem{Notation}[Definition]{Notation}
\newtheorem{Remark}[Definition]{Remark}
\newtheorem{Note}[Definition]{Note}

\def\b#1{\textcolor{blue}{#1}}

\def\al{\alpha}
\def\b{\beta}
\def\ga{\gamma}

\def\eps{\varepsilon}
\def\blackslug{\hbox{\hskip 1pt \vrule width 4pt height 8pt depth 1.5pt \hskip 1pt}}
\def\qed{\quad\blackslug\lower 8.5pt\null\par}

\def\RR{\mathbb{R}}

\def\p{\prime}
\def\pp{\prime\prime}
\def\ppp{\prime\prime\prime}
\def\erfc{\mathrm{erfc}}
\def\I{\mathcal{I}}
\def\J{\mathcal{J}}
\def\S{\mathcal{S}}
\def\E{\mathcal{E}}
\def\L{\mathcal{L}}
\def\N{\mathcal{N}}
\def\G{\mathcal{G}}
\def\lb{\left|}
\def\rb{\right|}
\def\lv{\left\|}
\def\rv{\right\|}
\def\Af{\mathbf{A}}
\def\Nf{\mathbf{N}}

\title{A quasi-solution approach to Blasius similarity equation with general boundary conditions}

\author{T. E. Kim$^{1}$} 
\address{The Mathematics Department\\The Ohio State University\\Columbus, OH 43210 USA}

\begin{document}
\vskip -0.2cm
\begin{abstract}
	A recently developed method \cite{Costinetal}, \cite{Dubrovin}, and \cite{BlasiusCT} is used to find an analytic approximate solution with rigorous error bounds to the classical Blasius similarity equation with general boundary conditions. This provides detailed proofs for the results reported in \cite{Kim}.
\end{abstract}
\vskip -2cm
\date{April 8, 2014}
\maketitle

\section{Introduction}
The classical similarity solution of Blasius to the boundary layer equation past a semi-infinite plate satisfies the two-point boundary value problem
\begin{equation} \label{eq:blasius}
f^{\ppp}(x) + f(x) f^{\pp} = 0 \quad \text{for} \quad x\in(0,\infty)
\end{equation}
with no-slip boundary conditions:
\begin{equation} \label{eq:bc}
f(0) = 0,\, f^{\p}(0) = 0,\, \text{and} \, \lim_{x\to+\infty} f^{\p}(x) = 1.
\end{equation}
One may consider \eqref{eq:blasius} with the following generalized boundary conditions:
\be{eq:gen_bc}
f(0) = \tilde{\al},\, f^{\p}(0) = \tilde{\ga},\, \text{and} \, \lim_{x\to+\infty} f^{\p}(x) = 1.
\ee
In \cite{Weyl}, using a transformation 
\be{eq:transform}
f(x) = a^{-1/2}F(a^{-1/2}x)
\ee
introduced by T\"{o}pfer \cite{Topfer}, it is shown that the boundary value problem \eqref{eq:blasius} and \eqref{eq:bc} can be written as the initial value problem
\be{eq:trans_blasius}
F^{\ppp}(x) + F(x)F^{\pp}(x) = 0 \quad \text{for} \quad x\in(0,\infty)
\ee
with initial conditions
\be{eq:ic}
F(0) = 0, \, F^{\p}(0) = 0, \, F^{\pp}(0) = 1.
\ee
At infinity, $\lim_{x\to\infty} F^{\p}(x) = a\in\RR^{+}$. Under the transformation, the generalized boundary conditions \eqref{eq:gen_bc} become
\be{eq:gen_ic}
F(0) = a^{1/2}\tilde{\al} \equiv \al, \, F^{\p}(0) = a\tilde{\ga} \equiv \ga, \, F^{\pp}(0) = 1.
\ee
The non-dimensionalized wall stress is given by 
\be{eq:wall-stress}
f^{\pp}(0) = a^{-3/2}
\ee
In this paper, the quasi-solution approach developed in \cite{BlasiusCT} is adopted to find an approximate analytic solution to the problem \eqref{eq:trans_blasius} with generalized initial conditions \eqref{eq:gen_ic} and to prove its rigorous error bounds. 

\section{Representation of a quasi-solution and main results}
For simplicity, we consider the initial value problem \eqref{eq:trans_blasius}--\eqref{eq:gen_ic} with $\ga = 0$ and $\al\in\J:= [-\frac{3}{50},\frac{3}{50}]$. (Throughout the rest of the paper, whenever \eqref{eq:gen_ic} is referred, this condition will be assumed.)
Through piecewise polynomial representations, other intervals in $\al$ can be examined in a similar fashion. Let
\be{eq:poly}
P(y;\b) = \sum_{i=0}^{13}\sum_{j=0}^{5}\frac{p_{i,j}}{(i+1)(i+2)(i+3)}\b^{j}y^{i}
\ee
where $p_{i,j}$ is the $(i+1,j+1)$-entry of the following matrix 
\be{}
\left[\begin{array}{cccccc}
{\frac{29589}{493148}} & -{\frac{9845}{82042}} & -{\frac{274}{40132715}} & {\frac{241}{11270972}} & -{\frac{422}{16143111}} & {\frac{308}{28130517}}\\
\noalign{\medskip}{\frac{15185}{1706376}} & -{\frac{17096}{473735}} & {\frac{36599}{968864}} & -{\frac{19441}{3418968}} & {\frac{6287}{892276}} & -{\frac{10649}{3570017}}\\
\noalign{\medskip}-{\frac{203116}{65155}} & -{\frac{3042}{970153}} & -{\frac{15440}{235863}} & {\frac{21239}{89058}} & -{\frac{114887}{372923}} & {\frac{5024}{37953}}\\
\noalign{\medskip}-{\frac{72804}{75433}} & {\frac{239497}{147253}} & {\frac{213995}{192583}} & -{\frac{110079}{28121}} & {\frac{1322305}{259224}} & -{\frac{80021}{35684}}\\
\noalign{\medskip}{\frac{106800}{43663}} & -{\frac{112122}{86717}} & -{\frac{155285}{19732}} & {\frac{525204}{17519}} & -{\frac{2029749}{49136}} & {\frac{391166}{20741}}\\
\noalign{\medskip}-{\frac{387344}{32609}} & {\frac{77473}{4402}} & {\frac{304475}{15867}} & -{\frac{3049469}{26658}} & {\frac{445437}{2501}} & -{\frac{568723}{6514}}\\
\noalign{\medskip}{\frac{3084825}{27611}} & -{\frac{1006071}{9319}} & {\frac{171511}{4286}} & {\frac{3723623}{24721}} & -{\frac{1097313}{2915}} & {\frac{1207261}{5453}}\\
\noalign{\medskip}-{\frac{2254258}{5883}} & {\frac{3595213}{9561}} & -{\frac{1049674}{2379}} & {\frac{2081034}{4399}} & {\frac{1013365}{19943}} & -{\frac{1249672}{5459}}\\
\noalign{\medskip}{\frac{1915077}{2126}} & -{\frac{3165632}{3527}} & {\frac{5196992}{3543}} & -{\frac{3429722}{1327}} & {\frac{3839299}{2153}} & -{\frac{2755673}{9363}}\\
\noalign{\medskip}-{\frac{2860297}{1927}} & {\frac{3706169}{2627}} & -{\frac{5245388}{1929}} & {\frac{1764108}{317}} & -{\frac{6522639}{1366}} & {\frac{1111693}{833}}\\
\noalign{\medskip}{\frac{281944}{179}} & -{\frac{3174435}{2257}} & {\frac{5003871}{1621}} & -{\frac{7633149}{1117}} & {\frac{6098777}{958}} & -{\frac{9281007}{4606}}\\
\noalign{\medskip}-{\frac{2506157}{2481}} & {\frac{2704059}{3157}} & -{\frac{8285683}{3873}} & {\frac{6455381}{1295}} & -{\frac{4186545}{863}} & {\frac{3106817}{1912}}\\
\noalign{\medskip}{\frac{2072736}{5813}} & -{\frac{1425478}{4881}} & {\frac{3778762}{4529}} & -{\frac{980233}{486}} & {\frac{3100252}{1537}} & -{\frac{4063417}{5821}}\\
\noalign{\medskip}-{\frac{1051227}{19699}} & {\frac{745495}{17357}} & -{\frac{1839247}{13071}} & {\frac{1844827}{5276}} & -{\frac{2241089}{6290}} & {\frac{3813801}{30274}}
\end{array}\right].
\ee

\begin{Definition}
For $\al \in \J$, define functions $a_0(\al), b_0(\al)$, and $c_0(\al)$ by
\begin{align}
a_{0}\left(\alpha\right) & =  {\frac{3221}{1946}}-{\frac{797}{603}}\,\alpha+{\frac{176}{289}}\,{\alpha}^{2}\\
b_{0}\left(\alpha\right) & =  -{\frac{2763}{1765}}+{\frac{761}{284}}\,\alpha-{\frac{194}{237}}\,{\alpha}^{2}\\
c_{0}\left(\alpha\right) & =  {\frac{377}{1613}}+{\frac{174}{1357}}\,\alpha+{\frac{937}{6822}}\,{\alpha}^{2}
\end{align}
and a subset $\S_\al$ of $\RR^3$ by
\be{}
S_\al = \left\{ (a,b,c)\in\RR^3 : \sqrt{ (a-a_0(\al))^2 + \frac{1}{4}(b-b_0(\al))^2 + \frac{1}{4}(c-c_0(\al))^2 } \le \rho_0 \right\}
\ee
where $\rho_0  := 5\times10^{-4}$.
\end{Definition}

\begin{Definition} \label{def:t}
Given $a,b,c\in\RR$ with $a> 0$, define
\begin{align}
t(x;a,b) & = \frac{a}{2} \left(x+\frac{b}{a}\right)^2 \\
q_0(t;c) & = 2c\sqrt{t}e^{-t}I_0(t) + c^2 e^{-2t} (2J_0(t) - I_0 (t) - I_0^2(t))\\
\intertext{in which}
I_0(t) & = 1 - \sqrt{\pi t}\,e^t\,\erfc(\sqrt{t})\quad \text{and} \quad J_0(t) =  I_0(2t)
\end{align}
where erfc denotes the complementary error function.
\end{Definition}

\begin{Note} \label{note:findingT}
Let $\al\in\J$ be arbitrary but fixed and let
\begin{alignat}{2} 
a_{l}&=a_{0}\left(\alpha\right)-\rho_{0}&\,,\quad
a_{r}&=a_{0}\left(\alpha\right)+\rho_{0}\,,\\
b_{l}&=b_{0}\left(\alpha\right)-2\rho_{0}&\,,\quad
b_{r}&=b_{0}\left(\alpha\right)+2\rho_{0}\,,\\
c_{l}&=c_{0}\left(\alpha\right)-2\rho_{0}&\,,\quad
c_{r}&=c_{0}\left(\alpha\right)+2\rho_{0}\,.
\end{alignat}
Suppose $\left(a,b,c\right)\in\mathcal{S}_{\alpha}$. Then it follows that 
$a\in[a_{l},a_{r}]$,
$b\in[b_{l},b_{r}]$,
and $c\in[c_{l},c_{r}]$.
Since $a_{0}(\al)$ and $b_{0}(\al)$ are quadratic in $\alpha$,
simple calculations show that $a \in [1.5, 1.75]$ and $b \in [-1.75, -1.4]$, which implies that $\frac{b}{a} \ge -1.17$. 
In particular, the function $t(x):=t(x;a,b)$ maps bijectively the interval $x\in [\frac{5}{2},\infty)$ onto the interval $t\in[t_m,\infty)$ where $t_m := t(\frac{5}{2})$. 
Moreover, since $a_l>0$ and $b_r<0$, 
\begin{equation}
\frac{a_{l}}{2}\left(\frac{5}{2}+\frac{b_{l}}{a_{l}}\right)^{2}
< t_{m} <
\frac{a_{r}}{2}\left(\frac{5}{2}+\frac{b_{r}}{a_{r}}\right)^{2}.
\end{equation}
Since $\al$ was chosen arbitrary, we conclude that $t_m\in(t_{m,l},t_{m,r})$ where
\begin{align}
t_{m,l} & = \inf_{\alpha\in\mathcal{J}}\bigg\{ \frac{a_{l}}{2}\left(\frac{5}{2}+\frac{b_{l}}{a_{l}}\right)^{2}\bigg\}=1.962257\cdots \\
t_{m,r} & = \sup_{\alpha\in\mathcal{J}}\bigg\{ \frac{a_{r}}{2}\left(\frac{5}{2}+\frac{b_{r}}{a_{r}}\right)^{2}\bigg\}=2.043219\cdots . 
\end{align}
So, provided that $a>0$, the domain $t\in[T,\infty)$ where $1.96\le T \le t_{m,l}$ corresponds to the domain $x\in [-\frac{b}{a}+\sqrt{\frac{2T}{a}},\infty)$ which is guaranteed to include $x\in[\frac{5}{2},\infty)$. 
\end{Note}

The theorem below provides an approximate analytic representation of solution $F_\al$ to \eqref{eq:trans_blasius} and \eqref{eq:gen_ic} with $\al\in\J$ and $\ga=0$. 

\begin{Theorem} \label{theorem}
Let $\alpha\in\J$ and $\ga = 0$. Then there exists a unique triple $(a,b,c)=(a(\al),b(\al),c(\al))\in\S_\al$ such that the function $F_{0,\al}$ defined by
\be{}
F_{0,\al} (x)  =
\begin{dcases}
\al+\frac{x^{2}}{2}+x^{3}P\left(\frac{2}{5}x;\frac{25}{3}\alpha+\frac{1}{2}\right)\,, & x\in\left[0,\frac{5}{2}\right]\\
ax+b+\sqrt{\frac{a}{2t(x)}}q_{0}(t(x);c)\,, & x\in\left(\frac{5}{2},\infty\right)
\end{dcases}
\ee
is a representation of the actual solution $F_\al$ to the initial value problem \eqref{eq:trans_blasius} and \eqref{eq:gen_ic} within small errors. More precisely, the error term $E_\al(x):=F_\al(x) - F_{0,\al}(x)$ satisfies on $\I:=[0,\frac{5}{2}]$
\be{eq:error_finite}
\Vert E_\al^{\pp}\Vert_{\infty,\I}\le4.8916\times10^{-6}\,,\, 
\Vert E_\al^{\p}\Vert_{\infty,\I}\le3.7474\times10^{-6}\,,\,
\Vert E_\al\Vert_{\infty,\I}\le7.4947\times10^{-6}
\ee
and for $x>\frac{5}{2}$ 
\begin{gather} \label{eq:error_far}
\left|E_\al^{\pp}(x)\right|\le5.4901\times10^{-4}t^{-1}e^{-3t}\,,\,
\left|E_\al^{\p}(x)\right|\le9.8179\times10^{-5}t^{-3/2}e^{-3t}\,,\\
\left|E_\al(x)\right|\le1.7558\times10^{-5}t^{-2}e^{-3t},\nonumber
\end{gather}
where $t = t(x;a,b)$.
\end{Theorem}

The proof of Theorem \ref{theorem} relies on the following three propositions. 

\begin{Proposition} \label{prop:finite}
For each $\al\in\J$, let $F_{1,\al}(x)$ be the solution of \eqref{eq:trans_blasius} and \eqref{eq:gen_ic} on $\I$. Then the error term $E_\al(x)\equiv F_{1,\al}(x)-F_{0,\al}(x)$ verifies the equation 
\begin{align}
\L[E_\al] :&=E_\al^{\ppp}+F_{0,\al}E_\al^{\pp}+F_{0,\al}^{\pp}E_\al \label{eq:E_ivp}\\
& =-F_{0,\al}^{\ppp}-F_{0,\al}F_{0,\al}^{\pp}-E_\al E_\al^{\pp},\nonumber \\
E_\al(0)&=E_\al^{\p}(0)=E_\al^{\pp}(0)=0. \label{eq:E_ic}
\end{align}
for $x\in\I$ and satisfies the bounds given in \eqref{eq:error_finite}.
\end{Proposition}

\begin{Proposition} \label{prop:far}
Let $T\ge1.96$. Given $(a,b,c)$ with $a>0$, $|c|\leq\frac{1}{4}$, in the domain $x\in[-\frac{b}{a}+\sqrt{\frac{2T}{a}},\infty)$, which corresponds to the domain $t = t(x;a,b) \in [T,\infty)$, there exists a unique solution to \eqref{eq:trans_blasius} in the form
\be{}
F_2(x;a,b,c)=ax+b+\sqrt{\frac{a}{2t}}q(t;c)
\ee
where the function $q(t;c)$ satisfies the condition 
\be{}
\lim_{t\to\infty}\frac{q(t;c)}{\sqrt{t}}=0.
\ee
Furthermore, the function $\mathcal{E}(t;c) \equiv q(t;c) - q_0(t;c)$ satisfies the following bounds for $t\in[T,\infty)$:
\begin{gather}
\left|\mathcal{E}(t;c)\right|\le1.6955\times10^{-4}\frac{e^{-3t}}{9t^{3/2}}\\
\left|\mathcal{E}^{\p}(t;c)-\frac{1}{2t}\mathcal{E}(t;c)\right|\le1.6955\times10^{-4}\frac{e^{-3t}}{3t^{3/2}}\\
\left|\sqrt{t}\mathcal{E}^{\pp}(t;c)-\frac{1}{2\sqrt{t}}\mathcal{E}^{\p}(t;c)+\frac{1}{2t^{3/2}}\mathcal{E}(t;c)\right|\le1.6955\times10^{-4}\frac{e^{-3t}}{t}.
\end{gather}
\end{Proposition}


\begin{Proposition} \label{prop:matching}
For each $\al\in\J$, there exists a unique triple $(a,b,c)\in\mathcal{S}_{\al}$ so that the functions $F_1$ and $F_2$ in the previous two propositions and their first two derivatives agree at $x=\frac{5}{2}$. 
\end{Proposition}

\noindent \textbf{The proof }of Theorem \ref{theorem} follows from Propositions \ref{prop:finite}--\ref{prop:matching} as follows: Proposition \ref{prop:finite} implies that for any $\al\in\J$, $F_{1,\al}(x) = F_{0,\al}(x) + E_\al(x)$ satisfies \eqref{eq:trans_blasius} and \eqref{eq:gen_ic} for $x\in\I$. Note that $F_{0,\al}$ satisfies the initial conditions $F_{0,\al}(0) = \al$, $F_{0,\al}^{\p}(0) = 0$, and $F_{0,\al}^{\pp}(0) = 1$. 

Proposition \ref{prop:far} implies that $F_2(x;a,b,c) = ax + b + \sqrt{\frac{a}{2t}}[q_0(t;c) + \E(t;c)]$, where $t = t(x;a,b)$, satisfies \eqref{eq:trans_blasius} in the domain of $x$ that includes $[\frac{5}{2},\infty)$ when $(a,b,c)\in\S_\al$.
Proposition \ref{prop:matching} ensures that both $F_1$ and $F_2$ solve the same ODE \eqref{eq:trans_blasius}. Furthermore, identifying $F_{0,\al}(x)$ and $E_\al(x)$ from Theorem \ref{theorem} for $x\in(\frac{5}{2},\infty)$ with $ax + b + \sqrt{\frac{a}{2t}}q_0(t;c)$ and $\sqrt{\frac{a}{2t}}\E(t;c)$ respectively, and relating $x$-derivatives to $t$-derivatives via $t = t(x;a,b)$, the error bounds \eqref{eq:error_far} follows from Proposition \ref{prop:far}. \hfill \qed
The proofs of Propositions \ref{prop:finite}--\ref{prop:matching} are presented in the following sections.

\section{Solution in the finite domain $\I=[0,\frac{5}{2}]$ and proof of Proposition \ref{prop:finite}}
The quasi-solution $F_0$ in the compact set $\I$ is obtained by fitting numerical solutions of \eqref{eq:trans_blasius} on $\I$ with high accuracy satisfying various initial conditions. More precisely, let $N=50$, $\al_k := -\frac{3}{50} + \frac{3}{25}\frac{k}{N}$ for $k=0,1,\dots,N$, and $\widehat{F_{0,k}}$ be the numerical solution of 
\begin{gather}
F^{\ppp}(x) + F(x)F^{\pp}(x) = 0,\quad \text{for} \quad x\in\I \\
F(0) = \al_k, \quad F^{\p}(0) = 0, \quad F^{\ppp}(0) = 1
\end{gather}
with absolute errors of the order $10^{-16}$. 
Since numerical differentiation is ill-conditioned, we project the the third derivative of the numerical solutions, rather than the solutions themselves, onto the subspace spanned by first several Chebyshev polynomials to obtain the set of $N$ approximate third derivatives written in the form
\be{}
P_{k}(x) = \sum_{n=0}^{M} c_{k,n}x^n,\quad k=0,\dots,N.
\ee
We then fit the coefficients $(c_{k,n})_{k=0}^{N}$ against $(\al_k)_{k=0}^{N}$ by degree 5 polynomial $c_n(\al)$ for all $n=0,\dots,M$ and write 
\be{}
P(x;\al) = \sum_{n=0}^{M} c_n(\al) x^n.
\ee{}
This is how the polynomial \eqref{eq:poly} is obtained. 

We seek to control the error term $E_\al(x)$ on the interval $\I$ uniformly in $\al\in\J$ by first estimating the size of the residual
\be{}
R_\al(x) = F_{0,\al}^{\ppp}(x) + F_{0,\al}(x)F_{0,\al}^{\pp}(x).
\ee
\begin{Notation}
In the following analysis, we will use two notations, for instance, $R_\al(x)$ and $R(x;\al)$ interchangeably. To be more precise, in places where we view the parameter $\al$ as another variable, we regard $R_\al(x)$ as a function $R(x;\al)$ of two variables $x$ and $\al$ on $\I \times \J$. In such a case, the derivatives with respect to $x$ will be denoted by prime notations, e.g., $\partial_x R(x;\al) = R^{\p}(x;\al)$; the derivatives with respect to $\al$ will be denoted using the partial derivative symbols, e.g., $\partial_\al R(x;\al)$.
\end{Notation}
We then invert the principal part of $\L[E]$ in Proposition \ref{prop:finite} by using initial conditions to obtain a nonlinear integral equation. The smallness of $R$ and careful bounds on the resolvents allow us to use a contractive mapping argument to draw the desired conclusion. 
\subsection{Estimating sizes of the the quasi-solution and the residual on $\I$} \label{subsec:estimation}
In this subsection, two methods are used to estimate sizes of the quasi-solution and its derivatives as well as the residual. 

\subsubsection{Estimation using local Taylor series expansion}
Since $F_{0,\al}(x)$, now viewed as $F_0(x;\al)$,  is a polynomial of $x$ and $\al$ of degree 16 and 5 respectively, $R$ is a polynomial of $x$ and $\al$ of degree 30 and 10 respectively:
\be{}
R(x;\al) = \sum_{m=0}^{30} \sum_{n=0}^{10} c_{m,n} \al^n x^m.
\ee
Based on how rapidly $R(x;\al)$ changes in $\I$ and in $\J$, we choose $\{x_k\}_{k=0}^{15} \in \I$ given by 
\[
\{ 0, 0.0625, 0.125, 0.25, 0.375, 0.5, 0.75, 1.0, 1.25, 1.4, 1.5, 1.75, 2, 2.25, 2.4, 2.5 \}.
\]
and $\{\al_l\}_{l=0}^{5} \in \J$ given by 
\[
\{ -0.06, -0.05, -0.02, 0.02, 0.05, 0.06 \}.
\]
We will show that $R$ is small in the sense that the norm given by 
\be{}
\lv R \rv_{\infty,\I\times\J} := \sup \left\{ \lb R(x;\al) \rb : x\in\I, \al\in\J \right\} 
\ee
is small. 

On each subregion $[x_{k-1},x_{k}] \times [\al_{l-1},\al_{l}]$ in $\I\times\J$, re-expand $R$ in the scaled variables $\tilde{x}_k$ and $\tilde{\al}_l$ where
\be{}
\al = \frac{\al_l + \al_{l-1}}{2}+\frac{\al_l - \al_{l-1}}{2}\tilde{\al}_l \quad \text{and} \quad x = \frac{x_k + x_{k-1}}{2}+\frac{x_k - x_{k-1}}{2}\tilde{x}_k.
\ee
Note that both $\tilde{\al}_l$ and $\tilde{x}_k$ are in $[-1,1]$. So we have
\begin{align} \label{eq:expansion}
R(x;\al) & = \sum_{m=0}^{30} \sum_{n=0}^{10} c_{m,n}^{(k,l)} \tilde{\al}_l^n\tilde{x}_k^m \nonumber\\
& = \sum_{m=0}^3 c_{m,0}^{(k,l)} \tilde{x}_k^m + \sum_{n=1}^3 c_{0,n}^{(k,l)}\tilde{\al}_l^n 
+ \sum_{m,n} c_{m,n}^{(k,l)} \tilde{\al}_l^n\tilde{x}_k^m
\end{align}
where the last sum is a double summation over all indices left out from the first two. Observe that the first two terms are single-variable cubic polynomials in $\tilde{x}_k$ and $\tilde{\al}_l$ respectively. So we can determine the maximum $M_{k,l}$ and the minimum $m_{k,l}$ of their sum in $\tilde{x}_k \in [-1,1]$ and $\tilde{\al}_l \in [-1,1]$ using calculus. The remaining term in \eqref{eq:expansion} is bounded by its $l^1$-norm:
\be{}
E_{k,l} := \sum_{m,n} \lb c_{m,n}^{(k,l)} \rb.
\ee
It follows that on $[x_{k-1},x_k]\times[\al_{l-1},\al_l]$,
\be{}
m_{k,l} - E_{k,l} \le R(x;\al) \le M_{k,l} + E_{k,l}\,.
\ee
The maximum and minimum over an arbitrary union of subregions are found by taking the minimum of $m_{k,l} - E_{k,l}$ and the maximum of $M_{k,l} + E_{k,l}$ over the appropriate indices $k$ and $l$. Note that, though elementary and tedious, these computations are executed easily and exactly with the aid of a computer algebra system since they only involve operations with rational numbers. 

Let $\I_k$ be defined by
\be{}
\I_1 = [0,1.25]\,, \quad \I_2 = [1.25,1.4]\,, \quad \I_3 = [1.4,2]\,, \quad \I_4 = [2,2.5]\,.
\ee
Using the method outlined above, we obtain estimates of the size of $R$ on subregions $\I_k\times \J$:
\begin{align} 
- 4.9058 \times 10^{-7} & \le R(x;\al) \le 5.1794 \times 10^{-7} \quad \text{on $\I_1\times\J$},\label{eq:RbndBegin}\\
- 8.4748 \times 10^{-8} & \le R(x;\al) \le 7.5413 \times 10^{-7} \quad \text{on $\I_2\times\J$},\\
 1.1011 \times 10^{-7} & \le R(x;\al) \le 1.3040 \times 10^{-6} \quad \text{on $\I_3\times\J$},\\
 4.9134 \times 10^{-7} & \le R(x;\al) \le 2.9344 \times 10^{-6} \quad \text{on $\I_4\times\J$} \label{eq:RbndEnd}.
\end{align}
This implies that $\lv R \rv_{\infty,\I \times \J} \le 2.9344 \times 10^{-6}$. The same method is used to estimate the size of $F_0$ on the subregions:
\begin{align}
-0.0601 & \le F_0(x;\al) \le 0.8004 \quad \text{on $\I_1\times\J$},\\
0.7157  &\le F_0(x;\al) \le 0.9753 \quad \text{on $\I_2\times\J$},\\
0.9039  &\le F_0(x;\al) \le  1.7938  \quad \text{on $\I_3\times\J$},\\
1.7819  &\le F_0(x;\al) \le  2.6220  \quad \text{on $\I_4\times\J$}\,.
\end{align}
The size of $F_0^{\p}$ is similarly estimated:
\begin{align}
-0.0001 & \le F_0^{\p}(x;\al) \le 1.1990 \quad \text{on $\I_1\times\J$},\\
1.1179  &\le F_0^{\p}(x;\al) \le 1.3091 \quad \text{on $\I_2\times\J$},\\
1.2134  &\le F_0^{\p}(x;\al) \le  1.6066  \quad \text{on $\I_3\times\J$},\\
1.4660  &\le F_0^{\p}(x;\al) \le  1.7036  \quad \text{on $\I_4\times\J$}\,.
\end{align}
Lastly, the estimates of $F_0^{\pp}$ are given:
\begin{align}
0.6770 & \le F_0^{\pp}(x;\al) \le 1.0144 \quad \text{on $\I_1\times\J$},\\
0.5927  &\le F_0^{\pp}(x;\al) \le 0.7778 \quad \text{on $\I_2\times\J$},\\
0.2599  &\le F_0^{\pp}(x;\al) \le  0.6890  \quad \text{on $\I_3\times\J$},\\
0.0881  &\le F_0^{\pp}(x;\al) \le  0.3099  \quad \text{on $\I_4\times\J$}\,.
\end{align}

\Remark{For any given $x\in \I$, $F_0, F_0^{\p}, F_0^{\pp}$ behave ``almost linearly" in $\al \in \J$. So we can still obtain fairly reasonable estimates without subdividing $\J$ as presented above.}

\subsubsection{Alternate method using Chebyshev polynomials}
Alternatively, we can find a bound on $\lv R \rv _{\infty,\I\times\J}$ by projecting it onto the orthogonal space of Chebyshev polynomials. To be more precise, take the Chebyshev expansions of the monomials $x^m$ and $\al^n$ on $\I$ and $\J$ respectively and write
\be{}
x^m = \sum_{i=0}^m p_{m,i}T_i(\tilde{x}) \text{  and  } \al^n = \sum_{j=0}^n q_{n,j}T_j(\tilde{\al}),
\ee
where $T_k$ is the Chebyshev polynomial of the first kind with degree $k$, $\tilde{x} = \frac{4x}{5}-1$, and $\tilde{\al}=\frac{50\al}{3}$. On substitution, we can rewrite $R$ as 
\begin{align}
R(x;\al) & = \sum_{i=0}^{30} \sum_{j=0}^{10} r_{i,j} T_j(\tilde{\al}) T_i(\tilde{x})\,,\\
\intertext{where}
r_{i,j} & = \sum_{m=i}^{30} \sum_{n=j}^{10} c_{m,n} p_{m,i} q_{n,j}\,.
\end{align}
Since $\lb T_k(y) \rb = \lb \cos(k \cos^{-1} y) \rb \le 1$ for $y \in [-1,1]$, it immediately follows that 
\be{eq:R_cheby}
\lv R \rv_{\infty,\I\times\J} \le \sum_{i=0}^{30} \sum_{j=0}^{10} \lb r_{i,j} \rb.
\ee
Using a computer algebra system, we obtain that $\lv R \rv_{\infty,\I\times\J} \le 3.5551 \times 10^{-6}$. Projecting $R$ to Chebyshev polynomials in each of the subregions $\I_k\times\J$ yields somewhat better bounds:
\begin{align}
\lv R \rv_{\infty,\I_1\times\J} \le 5.3776 \times 10^{-7}, &\qquad
\lv R \rv_{\infty,\I_2\times\J} \le 9.6144 \times 10^{-7}, \\
\lv R \rv_{\infty,\I_3\times\J} \le 1.5004 \times 10^{-6}, &\qquad
\lv R \rv_{\infty,\I_4\times\J} \le 2.9505 \times 10^{-6}.
\end{align}
Observe that these bounds are not as sharp as the ones obtained in \eqref{eq:RbndBegin}--\eqref{eq:RbndEnd}. Yet, this method is simpler and more easily adapted.
\subsection{Properties of some functions used in the subsequent sections.} \label{subsec:functions}
In this subsection, we will show that, for any fixed $\al \in \J$, each of the functions 
\be{}
G_1:=2F_0^{\pp} - 2F_0\,, \quad G_2:=F_0^{\pp} - 2F_0\,, \quad G_3:=F_0^{\pp} - 2F_0 + 1\,,
\ee
has a unique zero in the interval $\I$. 

We consider $G_3$ first. Based on the calculations above, this function is positive on $\I_1$. In addition, using the bounds of $R$, we see that its derivative $F_0^{\ppp} - 2F_0 = -F_0 F_0^{\pp} - 2F_0 +R$ is negative on $\I_2\cup \I_3 \cup \I_4$. Thus, $G_3$ has at most one root in $\I$. Now, the function $G_3(1.25;\al)$ is a polynomial in $\al$ of order 5. Applying the method given in Subsection \ref{subsec:estimation}, we obtain that 
\be{}
0.0781 \le G_3(1.25;\al) \le 0.3463\,\text{ whereas }
-0.3564 \le G_3(1.4;\al) \le - 0.1190,
\ee
for all $\al \in \J$, which means that for any given $\al\in\J$, the values of $G_3$ at $x=1.25$ and $x=1.4$ have the opposite signs. So by the intermediate value theorem, there exists a unique zero of $G_3$ between the two numbers. Similarly, we can show that there is a unique zero of $G_1$ in $\I$ between $x=1.15$ and $x=1.3$ and that $G_2$ has its only zero in $\I$ between $x=0.85$ and $x=1.05$.

\subsection{The error estimation using the energy method}
Let $\al \in \J$ be fixed and consider the linear (generally) inhomogeneous equation
\be{}
\L [\phi](x) := \phi^{\ppp}(x) + F_0(x;\al) \phi^{\pp}(x) + F_0^{\pp}(x;\al) \phi(x) = r(x)
\ee
over an arbitrary subinterval $[x_l,x_r] \subset \I$, with known initial conditions $\phi(x_l), \phi^{\p}(x_l)$, and $\phi^{\pp}(x_l)$. 
The solution to this equation is given by the standard variation of parameter formula:
\be{}
\phi(x) = \sum_{j=1}^{3} \phi^{(j-1)}(x_l)\Phi_{j,\al}(x) + \sum_{j=1}^{3}\Phi_{j,\al}(x)\int_{x_l}^{x} \Psi_{j,\al}(t) r(t) \,dt
\ee
where $\left \{ \Phi_{j,\al} \right \}_{j=1}^3$ form a 
fundamental set of solutions to $\mathcal{L_\al}[\phi]=0$ 
and $\left \{\Psi_{j,\al} \right \}_{j=1}^3$ are elements 
of the inverse of the fundamental matrix
constructed from the $\Phi_{j,\al}$ and their derivatives. (In what follows, for the sake of notational simplicity, we will suppress the $\al$-subscript but remember that these fundamental solutions depend on $\al$.) Since we seek to find the bounds on $\lv \phi \rv_\infty = \lv \phi \rv_{\infty,[x_l,x_r]}$, the precise expressions are unimportant. Rather, we proceed by differentiating $\phi$ twice using properties of 
$\Phi_j$ and $\Psi_j$ \footnote{In particular, 
$\sum_{j=1}^3 \Phi_j (x) \Psi_j (x) =0$, 
$\sum_{j=1}^3 \Phi_j^\prime (x) \Psi_j (x)=0$}
to obtain
\begin{equation}
\label{12.3.1}
\phi^{\prime \prime} (x) = 
\sum_{j=1}^3 \phi^{(j-1)} (x_l) \Phi_j^{\prime \prime} (x) 
+ \sum_{j=1}^3 \Phi_j^{\prime \prime} (x) \int_{x_l}^x \Psi_j (t) r(t) \,dt \,.
\end{equation}
Rewrite \eqref{12.3.1} by abstractly replacing the second term by an operator $\G$:
\begin{equation}
\label{12.4} 
\phi^{\prime \prime} (x) = \sum_{j=1}^3 \phi^{(j-1)} (x_l) \Phi_j^{\prime \prime} (x) + \mathcal{G} \left [ r \right ] (x).  
\end{equation}
From general properties of fundamental
matrix and its inverse for the linear ODEs with polynomial coefficients, 
$\mathcal{G}$ is a bounded linear operator 
on $\mathrm{C}([x_l, x_r])$; 
denote its norm by $M_\al$,
\begin{equation}
  \label{eq:defM}
  M_\al=\|\mathcal{G}\|\,.
\end{equation}
 Then, on the interval $[x_l, x_r]$, we have
\begin{equation}
\label{12.5}
\| \phi^{\prime \prime} \|_\infty \le  \sum_{j=1}^\infty M_{j,\al} \Big | \phi^{(j-1)} (x_l) \Big | 
+ M_\al \| r \|_{\infty}\quad \text{where} \quad M_{j,\al}
= \sup_{x \in [x_l,x_r]} 
\Big | \Phi_j^{\prime \prime} (x) \Big |. 
\end{equation}

To determine bounds on $M_{j,\al}$ and $M_\al$, we use the ``energy method": take the original ODE
\begin{equation}
\label{12.13}
\phi^{\prime \prime \prime} + F_0 \phi^{\prime \prime} + F_0^{\prime \prime} \phi = r\,,
\end{equation}
multiply it by $2 \phi^{\prime \prime}$, and then integrate from $x_l$ to $x$ using the known initial conditions to obtain
\begin{multline}
\label{12.14}
\left ( \phi^{\prime \prime} (x) \right )^2 = \left ( \phi^{\prime \prime} (x_l) \right )^2 \\
- \int_{x_l}^x \left \{ 2 F_0 (y;\al) \left ( \phi^{\prime \prime} (y) \right )^2  
+
2 F_0^{\prime \prime} (y;\al) \phi^{\prime \prime} (y) \phi (y) - 
2 \phi^{\prime \prime} (y) r(y) \right \}\, dy\,.
\end{multline}     

Note that we can express $\phi(x)$ in terms of $\phi^{\prime \prime} (x)$ by using integration by parts along with the known $\phi (x_l)$ and $\phi^\prime (x_l)$:
\begin{equation}
\label{12.15}
{\tilde \phi} (x) := 
\phi(x) - \phi (x_l) - (x-x_l) \phi^\prime (x_l) = \int_{x_l}^x (x-y) \phi^{\prime \prime} (y)\, dy.
\end{equation}
Using \eqref{12.15}, the equation \eqref{12.14} is now written as
\begin{multline}
\label{12.14.1}
\left ( \phi^{\prime \prime} (x) \right )^2 = \left ( \phi^{\prime \prime} (x_l) \right )^2 
- \int_{x_l}^x 2 F_0^{\prime \prime} (y;\al) 
\left [ \phi(x_l) + (y-x_l) \phi^{\prime} (x_l) \right ] \phi^{\prime \prime} (y)\, dy \\
- \int_{x_l}^x \left \{ 2 F_0 (y;\al) \left ( \phi^{\prime \prime} (y) \right )^2  
+
2 F_0^{\prime \prime} (y;\al) \phi^{\prime \prime} (y) {\tilde \phi} (y) - 
2 \phi^{\prime \prime} (y) r(y) \right \}\, dy.
\end{multline}  
Since the ODE of our interest is linear, we may consider separately the following cases to determine the bounds of $M_j$ and $M$:
\begin{itemize}
\item $r =0$, $\phi(x_l) =1$, $\phi^{\p} (x_l) =0$, $\phi^{\pp} (x_l) =0$;
\item $r =0$, $\phi(x_l) =0$, $\phi^{\p} (x_l) =1$, $\phi^{\pp} (x_l) =0$;
\item $r =0$, $\phi(x_l) =0$, $\phi^{\p} (x_l) =0$, $\phi^{\pp} (x_l) =1$;
\item $r \neq0$, $\phi(x_l) =0$, $\phi^{\p} (x_l) =0$, $\phi^{\pp} (x_l) =0$.
\end{itemize}
Using the simple inequality $-2ab \le a^2 + b^2$, the relation \eqref{12.15}, and Gronwall's inequality, it is shown (see \cite{BlasiusCT} for details) that
\begin{align}
M_{1,\al} 
& \le ((F_0^{\p}(x_r;\al) - F_0^{\p}(x_l;\al))^{1/2} \exp \left[ \frac{1}{2} \int_{x_l}^{x_r} Q_1(y;\al)\,dy \right]\,, \label{eq:SupBegin}\\
M_{2,\al}
& \le \left( \int_{x_l}^{x_r} (y-x_l)^2F_0^{\pp}(y;\al)\,dy \right)^{1/2} \exp \left[ \frac{1}{2} \int_{x_l}^{x_r} Q_1(y;\al)\,dy \right]\,,\\
M_{3,\al}
& \le \exp \left[ \frac{1}{2} \int_{x_l}^{x_r} Q_2(y;\al)\,dy \right]\,, \\
M_{\al}
& \le (x_r - x_l)^{1/2} \exp \left[ \frac{1}{2} \int_{x_l}^{x_r} Q(y;\al)\,dy \right]\,,\label{eq:SupEnd}
\end{align}
where
\begin{align}
Q_1(x;\al) & =
	\begin{dcases}
	\frac{(x-x_l)^4}{4} F_0^{\pp}(x;\al)+G_1(x;\al) & \text{if } G_1(x;\al) > 0\\
	\frac{(x-x_l)^4}{4} F_0^{\pp}(x;\al) & \text{if } G_1(x;\al) \le 0,
	\end{dcases} \\
Q_2(x;\al) & =
	\begin{dcases}
	\frac{(x-x_l)^4}{4} F_0^{\pp}(x;\al)+G_2(x;\al) & \text{if } G_2(x;\al) > 0\\
	\frac{(x-x_l)^4}{4} F_0^{\pp}(x;\al) & \text{if } G_2(x;\al) \le 0,
	\end{dcases} \\
Q(x;\al) & =
	\begin{dcases}
	\frac{(x-x_l)^4}{4} F_0^{\pp}(x;\al)+G_3(x;\al) & \text{if } G_3(x;\al) > 0\\
	\frac{(x-x_l)^4}{4} F_0^{\pp}(x;\al) & \text{if } G_3(x;\al) \le 0.
	\end{dcases}
\end{align}
(See Section \ref{subsec:functions} for the definition of $G_j$'s.)

Now, using the estimation method introduced in Subsection \ref{subsec:estimation}, we can show that the $\al$-derivatives of the following functions
\be{}
F_0^{\p}(x;\al) - F_0^{\p}(x_l;\al)\,,\,
F_0^{\pp}(x;\al)\,,\,
G_1(x;\al)\,,\,
G_2(x;\al)\,,\,
G_3(x;\al)\,,
\ee
are all negative for any given $x\in\I$. This implies that  these functions are decreasing in $\al$ on the interval $\J$ and thus they attain the maximal values at $\al=-\frac{3}{50}$ for any $x\in\I$. This allows us to uniformly bound $M_{j,\al}$ and $M_\al$ by $M_j$ and $M$ respectively. The results are summarized in Table \ref{tab:suprema}.


\begin{table}
\caption{The bounds of various suprema on subregions $\I_k\times\J$.}
\label{tab:suprema}
\begin{center}
\begin{tabular}{|c||c|c|c|c|}
\hline                              
 & $M$ & $M_{1}$ & $M_{2}$ & $M_{3}$\tabularnewline
\hline
\hline
$\I_{1}\times \J$ & $3.1930$ & $3.0482$ & $2.1323$ & $1.5886$\tabularnewline
$\I_{2}\times \J$ & $0.3912$ & $0.3323$ & $0.0284$ & $1.0001$\tabularnewline
$\I_{3}\times \J$ & $0.7762$ & $0.5465$ & $0.1701$ & $1.0020$\tabularnewline
$\I_{4}\times \J$ & $0.7077$ & $0.3120$ & $0.0775$ & $1.0008$\tabularnewline
\hline
\end{tabular}
\par\end{center}
\end{table}

\subsection{The existence of solution and the error estimates}
Using the results from the previous subsection, we can now not only show the existence and the uniqueness of the error $E_\al$ in the decomposition of the solution \be{eq:finite_soln}
F_\al(x) = F_{0,\al}(x) + E_\al(x),
\ee
but also show that it is small uniformly in $\al$ on an arbitrary subinterval $[x_l,x_r]\subset \I$. Suppose $E_\al(x_l),E_\al^{\p}(x_l)$, and $E_\al^{\pp}(x_l)$ are known. Then on $[x_l,x_r]$, $E_\al$ satisfies
\be{}
\L_\al[E_\al]=-E_\al E_\al^{\pp} - R_\al
\ee
where $R_\al = F_\al^{\ppp} + F_\al F_\al^{\pp}$. As in \eqref{12.4}, this equation is equivalent to the integral equation
\be{eq:E_int}
E_\al^{\pp} = \sum_{j=1}^{3} E_\al^{(j-1)}(x_l) \Phi_\al^{\pp}(x) - \G[R_\al](x) - \G[E_\al E_\al^{\pp}](x) =: \N[E_\al^{\pp}](x)
\ee
The following lemma which we directly quote from \cite{BlasiusCT} shows that the equation \eqref{eq:E_int} has a unique solution using a contractive mapping argument and provides an error estimate:
\begin{Lemma} \label{lem:finite}
Let $\al\in\J$ be fixed and assume that for some $\eps>0$ we have 
\begin{gather}
M\left(\lb E_\al(x_l) \rb + (x_r - x_l) \lb E_\al^{\p}(x_0) \rb \right) (1+\eps) + \frac{1}{2}(x_r-x_l)^2MB_0(1+\eps)^2 < \eps\,, \\
M\left(\lb E_\al(x_l) \rb + (x_r - x_l) \lb E_\al^{\p}(x_0) \rb \right) + (x_r-x_l)^2MB_0(1+\eps) < 1\,,
\intertext{where}
B_0 = M\lv R_\al \rv_{\infty,[x_l,x_r]} + \sum_{j=1}^3 M_j \lb E_\al^{(j)}(x_l) \rb.
\end{gather}
Then there exists a unique solution $E_\al^{\pp}$ of \eqref{eq:E_int} in a ball of radius $B_0(1+\eps)$ in the space $\mathrm{C}([x_l,x_r])$ equipped with the sup-norm $\lv \cdot \rv_{\infty,[x_l,x_r]}$. 
\end{Lemma}
We refer readers to \cite{BlasiusCT} for proof. 

\subsection{End of proof of Proposition \ref{prop:finite}}
Starting from $\I_1=[0,1.25]$ with the known initial conditions $E_\al(0)=E_\al^{\p}(0)=E_\al^{\pp}(0)=0$, it is verified that the lemma applies to all the subintervals $\I_k$'s and yields small error bounds as shown in Table \ref{tab:error}. Hence we conclude that $E_\al$ satisfies the equation \eqref{eq:E_ivp}--\eqref{eq:E_ic} with the bounds given in Theorem \ref{theorem}.\hfill \qed

\begin{table}
\caption{The error estimates on subintervals}
\label{tab:error}
\begin{center}
\begin{tabular}{|c|c|c|c|c|c|}
\hline 
 & $B_{0}$ & $\eps$ & $\Vert E\Vert_{\infty,\mathcal{I}_{j}}$ & $\Vert E^{\prime}\Vert_{\infty,\mathcal{I}_{j}}$ & $\Vert E^{\prime\prime}\Vert_{\infty,\mathcal{I}_{j}}$\tabularnewline
\hline 
\hline 
$\mathcal{I}_{1}$ & $1.6538\times10^{-6}$ & $5\times10^{-6}$ & $1.6538\times10^{-6}$ & $2.0673\times10^{-6}$ & $1.2921\times10^{-6}$\tabularnewline
\hline 
$\mathcal{I}_{2}$ & $2.4371\times10^{-6}$ & $7\times10^{-7}$ & $2.4371\times10^{-6}$ & $3.6556\times10^{-7}$ & $1.6296\times10^{-6}$\tabularnewline
\hline 
$\mathcal{I}_{3}$ & $4.3873\times10^{-6}$ & $3\times10^{-6}$ & $4.3873\times10^{-6}$ & $2.6324\times10^{-6}$ & $2.6386\times10^{-6}$\tabularnewline
\hline 
$\mathcal{I}_{4}$ & $7.4947\times10^{-6}$ & $4\times10^{-6}$ & $7.4947\times10^{-6}$ & $3.7474\times10^{-6}$ & $4.8916\times10^{-6}$\tabularnewline
\hline 
\end{tabular}
\par\end{center}
\end{table}

\section{Solution in $t\ge T \ge 1.96$ for $a>0$, $\lb c\rb <\frac{1}{4}$ and proof of Proposition \ref{prop:far}}

The construction of quasi-solution $F_{0}$ for $x \in [\frac{5}{2}, \infty )$ relies on large $x$ asymptotics, which as it turns out, gives a desirably accurate solution  
in the entire interval. For the Blasius solution, it is known that any solution with $\lim_{x \rightarrow \infty} F^{\p}(x) = a > 0$ must have the representation 
\be{eq:far_soln} 
F(x) = a x+ b + G(x) 
\ee
where $G(x)$ is exponentially small in $x$ for large $x$. Indeed, through change of
variable $t=t(x;a,b)$ given in Definition \ref{def:t} and $G(t) = \sqrt{\frac{a}{2 t}} q(t)$ with $q$ satisfying
\be{eq:q_eq}
{\frac {d^{3}}{d{t}^{3}}}q + \left ( 1 +  
\frac {q}{2t} \right )  {\frac {d^{2}}{d{t}^{2}}}q
+ \left( -\frac {1}{2t}+\frac{3}{4 t^2} -
\frac {q}{4 t^{2}} \right) \frac {dq}{dt} 
+ \left ( \frac{1}{2 t^2} -\frac{3}{4 t^3} \right ) q  
 +\frac{q^2}{4 t^3} = 0 
\ee
and from a general theory \cite{Duke},\footnote{Though the non-degeneracy condition 
stated in \cite{Duke} does not hold, a small modification leads to the same result.}
it may be deduced
that small solutions $q$ must have the convergent series representation
\be{eq:q_series}
q (t) = \sum_{n=1}^\infty \xi^n Q_n (t)\,,\, {\rm where}~ \xi= \frac{c e^{-t}}{\sqrt{t}}  
\ee
where the equations for $Q_n$ may be deduced by plugging in \eqref{eq:q_eq} into \eqref{eq:q_series} and equating different powers of $\xi$. With appropriate matching
at $\infty$, one obtains 
\be{eq:Q12}
Q_1 (t) = 2 t I_0 (t)\quad \text{and} \quad Q_2 (t) = - t I_0(t) - t I_0(t)^2 + 2 t J_0(t)
\ee
where
\begin{align}
I_0(t) & = 1 - \sqrt{\pi t} e^{t} {\rm erfc} (\sqrt{t}) = \frac{1}{2} \int_0^\infty \frac{e^{-st}}{(1+s)^{3/2}} ds\,, \\
J_0(t) & = 1 - \sqrt{2 \pi t} e^{2 t} {\rm erfc} (\sqrt{2 t}) = \frac{1}{4} \int_0^\infty 
\frac{e^{-st}}{(1+s/2)^{3/2}} ds\,.
\end{align}
The two term truncation of \eqref{eq:q_series} proved adequate to determine an accurate quasi-solution in an $x$-domain that corresponds to $t \ge 1.96$ if $|c| \le \frac{1}{4}$ to within the quoted accuracy. Note that the solution is only complete after determining $(a, b, c)$ through 
matching of $F_\al$, $F_\al^\p$ and $F_\al^{\pp}$ at $x=\frac{5}{2}$.
Since $(a, b, c)$ only needs to be restricted to some small neighborhood of $(a_0(\al), b_0(\al), c_0(\al))$ to accomplish matching (see Proposition \ref{prop:matching}), the restriction 
$t \ge 1.96$ is seen to include $x \ge \frac{5}{2}$ as shown in Note \ref{note:findingT}. Furthermore, the restriction $|c| \le \frac{1}{4}$ in Proposition \ref{prop:matching} is appropriate for the quoted error estimates in $x \ge \frac{5}{2}$ in Theorem \ref{theorem}.
We decompose 
\be{eq:q} 
q(t) = q_0 (t) + \mathcal{E} (t)\,,
\ee
where 
\be{eq:q0}
q_0 (t) = \frac{c e^{-t}}{\sqrt{t}} Q_1 (t) + \frac{c^2 e^{-2t}}{t} Q_2 (t)\,. 
\ee
\begin{Note}
The functions $q$ and $q_0$ (as well as some others to be introduced later) are dependent on $c$, but for the simplicity of notation, it will be suppressed in the current section.
\end{Note}
On substituting in \eqref{eq:q_eq}, we obtain a nonlinear integral equation for $\E$:
\begin{multline} \label{eq:EpsilonDE}
\E^{\ppp} + \left(1+\frac{q_0}{2t} \right)\E^{\pp} + \left(-\frac{1}{2t} + \frac{3}{4t^2} - \frac{q_0}{4t^2} \right)\E^{\p} \\
 + \left(\frac{1}{2t^2} - \frac{3}{4t^3} + \frac{q_0^{\pp}}{2t} - \frac{q_0^{\p}}{4t^2} + \frac{q_0}{2t^3}\right)\E = -\frac{\E}{2t}\E^{\pp} + \frac{\E}{4t^2}\E^{\p} - \frac{\E^2}{4t^3} - R,
\end{multline}
where the remainder $R=R(t)$ is given by 
\begin{multline} \label{eq:Rfar}
R = 
{\frac {d^{3}}{d{t}^{3}}}q_0 + \left ( 1 +  
\frac {q_0}{2t} \right )  {\frac {d^{2}}{d{t}^{2}}}q_0 \\
+ \left( -\frac {1}{2t}+\frac{3}{4 t^2} -
\frac {q_0}{4 t^{2}} \right) \frac {d}{dt} q_0
+ \left ( \frac{1}{2 t^2} -\frac{3}{4 t^3} \right ) q_0  
 +\frac{q_0^2}{4 t^3}\,.
\end{multline}
Using the auxiliary function 
\be{eq:h}
h(t) = e^t \left(\sqrt{t}\E^{\pp}(t) - \frac{\E^{\p}(t)}{2\sqrt{t}} + \frac{\E(t)}{2t^{3/2}}\right)
\ee
which is related to $\E$ by
\be{eq:Epsilon}
\mathcal{E} (t) = \sqrt{t}
\int_{\infty}^t \frac{ds}{\sqrt{s}} \int_{\infty}^s \frac{e^{-\tau}}{\sqrt{\tau}} h(\tau)\, d\tau,
\ee
the equation \eqref{eq:EpsilonDE} is now written as 
\be{eq:hDE}
h^{\p} = -\frac{q_0 e^t}{2t}h + e^t B\E - \frac{\E}{2t}h - \sqrt{t}e^t R, 
\ee
where
\be{}
B(t) = - \frac{q_0^{\pp}(t)}{2t^{1/2}} + \frac{q_0^{\p}(t)}{4t^{3/2}} - \frac{q_0(t)}{4t^{5/2}}\,.
\ee
This equation can be rewritten in an integral form
\begin{multline} \label{eq:h_int}
h(t) = h_0(t) - \int_{\infty}^{t} \frac{q_0(\tau)e^{\tau}}{2\tau} h(\tau)\,d\tau \\
+ \int_{\infty}^{t} e^{\tau}B(\tau)\E(\tau)\,d\tau 
-\int_{\infty}^{t} \frac{\E(\tau)}{2\tau}h(\tau)\,d\tau =: \N[h](t)\,,
\end{multline}
where 
\be{}
h_0(t) = -\int_{\infty}^{t} \sqrt{\tau}e^{\tau}R(\tau)\,d\tau.
\ee
A contractive mapping argument in a small ball inside the Banach space of $\mathrm{C}([T,\infty))$ equipped with the weighted norm 
\be{defn:norm}
\lv h \rv := \sup_{t\ge T} te^{2t}\lb h(t) \rb
\ee
is possible by utilizing the smallness of the residual  $R=R (t)$ as shown in the following proposition:
\begin{Proposition}
For $\lb c\rb \le \frac{1}{4}$, $\eps < 0.03$, and $T\ge1.96$, there exists a unique solution to the integral equation \eqref{eq:h_int} in a ball of radius $(1+\eps)\lv h_0\rv$, implying that $\lv h\rv \le (1+\eps)\lv h_0 \rv \le 1.6955 \times 10^{-4}$.
\end{Proposition}
This proposition is also found in Costin and Tanveer \cite{BlasiusCT}, the only difference being $T\ge1.96$ here whereas $T \ge 1.99$ in their paper. The proof is omitted here. The error bounds given in Theorem \ref{theorem} follows immediately from: 
\begin{Lemma}
For $a>0$, $\lb c\rb\le\frac{1}{4}$, and $t\ge T\ge1.96$, the function $\E$ satisfies the following bounds:
\begin{align}
\lb \sqrt{\frac{a}{2t}} \E(t) \rb & \le  \sqrt{\frac{a}{2}} \frac{1}{9} t^{-2} e^{-3t} \lv h\rv\,, \\
\lb \frac{d}{dx} \sqrt{\frac{a}{2t}} \E(t) \rb & \le \frac{a}{3} t^{-3/2} e^{-3t} \lv h\rv\,,\\
\lb \frac{d^2}{dx^2} \sqrt{\frac{a}{2t}} \E(t) \rb & \le \sqrt{2}a^{3/2} t^{-1} e^{-3t} \lv h\rv\,,
\end{align}
where $t=t(x;a,b)=\frac{a}{2}\left(x+\frac{b}{a}\right)^2$. 
\end{Lemma}
\begin{proof}
Note that by the definition of the weighted norm $\lv \cdot \rv$ given in \eqref{defn:norm},
\be{eq:norm_ineq}
\lb \int_\tau^\infty s^{-1/2}e^{-s} h(s) \, ds \rb \le \frac{1}{3}\tau^{-3/2}e^{-3\tau}\lv h\rv.
\ee
Using \eqref{eq:Epsilon} and the above inequality, 
\be{}
\lb \E(t) \rb = 
\lb \int_{\infty}^t s^{-1/2} \int_{\infty}^s \tau^{-1/2} e^{-\tau} h(\tau)\, d\tau\,ds\rb \le \frac{1}{9}t^{-3/2}e^{-3t}\lv h\rv
\ee
and from this the first inequality follows immediately. To see the second statement, we note from \eqref{eq:Epsilon} that 
\be{}
\frac{d}{dx}\sqrt{\frac{a}{2t}}\E(t) = a\left(\E^{\p}(t) - \frac{1}{2t}\E(t)\right)
= a \int_\infty^t \tau^{-1/2} e^{-\tau} h(\tau)\,d\tau,
\ee
and use the inequality \eqref{eq:norm_ineq}. The last one follows from checking that 
\be{}
\frac{d^2}{dx^2}\sqrt{\frac{a}{2t}}\E(t) = \sqrt{2}a^{3/2}e^{-t}h(t)\,.
\ee
and using the definition of the norm. \hfill
\end{proof}
This leads to the proof of Proposition \ref{prop:far}.

\section{Matching of solutions and proof of Proposition \ref{prop:matching}}
Let $\al\in\J$ and $(a,b,c)\in \S_\al$. In order for the two representations of the solution, \eqref{eq:finite_soln} and \eqref{eq:far_soln}, to coincide at $x = \frac{5}{2}$ we match them and their first two derivatives at the point. Let $t_m=t(\tfrac{5}{2};a,b)$. Then by \eqref{eq:q0} and \eqref{eq:Epsilon}, we get
\begin{align} \label{N1}
a  = F_\al^\prime (\tfrac{5}{2}) 
- a \left ( q_0^\prime (t_m; c) - \frac{q_0 (t_m; c)}{2t_m} \right ) -  
a \int_{\infty}^{t_m} \frac{e^{-\tau}}{\sqrt{\tau}} h (\tau; c)\,d\tau 
=:N_1 (a, b, c) 
\end{align}
\begin{multline} \label{N2}
b  = F_\al(\tfrac{5}{2}) - \frac{5}{2} N_1(a,b,c) - \sqrt{\frac{a}{2 t_m}} q_0 (t_m; c)  \\
 - \sqrt{\frac{a}{2}} \int_{\infty}^{t_m} \tau^{-1/2} \int_{\infty}^\tau s^{-1/2} e^{-s} h(s; c)\,ds 
:=N_2 (a, b, c) 
\end{multline}
\begin{align} \label{N3}
c  =  \frac{1}{\sqrt{2} a^{3/2}}
\left [ V (t_m; c) + \frac{1}{c} h (t_m; c) \right ]^{-1}  
e^{t_m} F_\al^{\prime \prime} (\tfrac{5}{2}) =: N_3 (a, b, c) \qquad \qquad \qquad
\end{align}
where 
\be{}
V(t;c) = -\frac{2}{c} t e^t B(t;c).
\ee

\begin{Definition}
\label{defn:5}
We define $\mathbf{A} = (a,\tfrac{1}{2}b,\tfrac{1}{2}c)$ and 
\be{}
\Nf[\Af]=(N_1(a,b,c),\tfrac{1}{2}N_2(a,b,c),\tfrac{1}{2}N_3(a,b,c))\,.
\ee
For each $\al\in\J$, define 
\be{}
\mathbf{A}_{0,\al} = (a_0(\al),\tfrac{1}{2}b_0(\al),\tfrac{1}{2}c_0(\al)).
\ee
Define also
\be{}
\mathcal{S}_{\Af,\al} = \left \{ 
\| {\bf A} - {\bf A}_{0,\al} \|_2 \le \rho_0=5 \times 10^{-4} \right \}
\ee
where $\|. \|_2$ is the Euclidean norm and let
\begin{equation}
\label{17.3.0}
{\bf J} = \frac{\partial{\bf N}}{\partial {\bf A}} =
{\begin{pmatrix} \partial_a N_1 & 2 \partial_b N_1  & 2 \partial_c N_1 \cr
         \frac{1}{2} \partial_a N_2 &  \partial_b N_2 & \partial_c N_2 \cr
         \frac{1}{2} \partial_a N_3 &  \partial_b N_3 & \partial_c N_3 
\end{pmatrix}}
\end{equation}
be the Jacobian. Let $\|{\bf J} \|_2$ denote the
$l^2 $ (Euclidean) norm of the Jacobian matrix:
\begin{multline}
\label{17.3.0.0}
\| J \|^2_2 = 
\left ( \partial_a N_1 \right )^2 + 4 \left (\partial_b N_1\right )^2  + 4 \left ( \partial_c N_1 \right )^2 \\
+\frac{1}{4} \left ( \partial_a N_2 \right )^2 
+\left (\partial_b N_2\right )^2  + \left ( \partial_c N_2 \right )^2
+\frac{1}{4} \left ( \partial_a N_3 \right )^2 
+\left (\partial_b N_3\right )^2  + \left ( \partial_c N_3 \right )^2.
\end{multline}
\end{Definition}

\begin{Note}
$\Af \in \S_{\Af,\al}$ implies that $(a,b,c)\in\S_\al$. The system of equations \eqref{N1}--\eqref{N3} is now succinctly written as 
\be{}
\Af = \Nf [\Af].
\ee
\end{Note}

\begin{Lemma} \label{lem:matching}
Let $\al\in\J$. Suppose that there exists some $\beta \in (0, 1)$ satisfying 
\begin{align} 
\| {\bf A}_{0,\al} - {\bf N} [{\bf A}_{0,\al}]\|_2 & \le (1-\beta) \rho_0 \,, \label{contract1}\\
\sup_{{\bf A} \in \mathcal{S}_{\Af,\al}} \| {\bf J} \|_2\ & \le \beta \,. \label{contract2}
\end{align}
Then the equation ${\bf A} = \mathbf{N} [ {\bf A} ]$ has a unique solution in $\S_{\Af,\al}$.
\end{Lemma}

\begin{proof}
Fix an $\al\in\J$ and let $\Af \in \S_{\Af,\al}$. By the mean-value theorem, 
\begin{multline}
\| {\bf N} [{\bf A} ] - {\bf A}_{0,\al} \|_2 
\le \| {\bf N} [{\bf A}] - {\bf N} [{\bf A}_{0,\al} ] \|_2 + 
\| {\bf N} [{\bf A}_{0,\al}] - {\bf A}_{0,\al} \|_2  \\
\le \| {\bf J} \|_2 \rho_0  + \rho_0 (1-\beta)  
\le \rho_0.
\end{multline}
Moreover, if ${\bf A}_1, {\bf A}_2 \in \mathcal{S}_{\Af,\al}$,
\be{}
\| {\bf N} [{\bf A}_1 ] 
- {\bf N} [{\bf A}_2 ] \|_2 
\le \|{\bf J} \|_2 \| {\bf A}_1 - {\bf A}_2 \|_2  
\le \beta \|{\bf A}_1 - {\bf A}_2 \|_2 
\ee
This implies that the map $\Nf$ maps the ball $\S_{\Af,\al}$ back to itself and is contractive there. Hence, by the Banach space fixed point theorem, the conclusion follows. \hfill
\end{proof}

The Proposition \ref{prop:matching} follows from Lemma \ref{lem:matching} once we show that the conditions \eqref{contract1} and \eqref{contract2} are satisfied for any $\al\in\J$. Following the procedures outlined in \cite{BlasiusCT}, it is not difficult to show that $\beta \le 0.8381$ and that $\lv \Af_{0,\al}-\Nf[\Af_{0,\al}]\rv_2 \le 4.1443 \times 10^{-5} \le (1-\beta)\rho_0$ for any $\al\in\J$. This completes the proof of Proposition \ref{prop:matching}.\hfill \qed


\vfill 
\eject
\end{document}